\documentclass[english]{article}
\usepackage[T1]{fontenc}
\usepackage[latin1]{inputenc}
\usepackage{geometry}
\usepackage{mathrsfs}
\usepackage{amsmath}
\usepackage{amssymb}
\usepackage{graphicx}

\usepackage{amsthm}

\usepackage{mathrsfs}

\usepackage{amsfonts}

\usepackage{epsfig}

\usepackage{bm}

\usepackage{mathrsfs}

\usepackage{enumerate}

\usepackage{subfig}\usepackage[all]{xy}

\newtheorem{lem}{Lemma}[section]
\newtheorem{rem}{Remark}[section]
\newtheorem{prop}{Proposition}[section]

\newtheorem{ass}{Assumption}[section]
\newcounter{hypA}
\newenvironment{hypA}{\refstepcounter{hypA}\begin{itemize}
  \item[({\bf A\arabic{hypA}})]}{\end{itemize}}

\newcommand{\dif}{\mathrm{d}}

\newcommand{\bbE}{\mathbb{E}}

\newcommand{\bbN}{\mathbb{N}}

\newcommand{\cO}{\mathcal{O}}

\newcommand{\cI}{{\cal I}}

\usepackage{xspace}
\usepackage{tabu}
\usepackage{booktabs}

\begin{document}


\begin{center}

{\Large \textbf{A Multi-Index Markov Chain Monte Carlo Method}}
%
%
BY AJAY JASRA$^{1}$, KENGO KAMATANI$^{2}$, KODY LAW$^{3}$ \& YAN ZHOU$^{1}$


{\footnotesize $^{1}$Department of Statistics \& Applied Probability,
National University of Singapore, Singapore, 117546, SG.}
{\footnotesize E-Mail:\,}\texttt{\emph{\footnotesize staja@nus.edu.sg; stazhou@nus.edu.sg}}\\
{\footnotesize $^{2}$Graduate School of Engineering Science, Osaka University, Osaka, 565-0871, JP.}
{\footnotesize E-Mail:\,}\texttt{\emph{\footnotesize kamatani@sigmath.es.osaka-u.ac.jp}}\\
{\footnotesize $^{3}$Computer Science and Mathematics Division,
Oak Ridge National Laboratory,  Oak Ridge, 37934 TN, USA.}
{\footnotesize E-Mail:\,}\texttt{\emph{\footnotesize kodylaw@gmail.com}}
\end{center}


\abstract{In this article we consider computing expectations w.r.t.~probability laws associated
to a certain class of stochastic systems. In order to achieve such a task, one must not
only resort to numerical approximation of the expectation, but also to a biased discretization of the
associated probability. We are concerned with the situation for which the discretization is required in multiple
dimensions, for instance in space-time. In such contexts, it is known that the multi-index Monte Carlo (MIMC)
method of \cite{mimc} can improve upon i.i.d.~sampling from the most accurate approximation of the probability law.
Through a non-trivial modification of the multilevel Monte Carlo (MLMC) method, 
this method can reduce the work to obtain a given level of error, 
relative to i.i.d.~sampling and relative even to MLMC. 
In this article we consider the case when such probability laws are too complex to be sampled
independently, for example a Bayesian inverse problem where evaluation of the likelihood requires
solution of a partial differential equation (PDE) model which needs to be approximated at finite resolution.
We develop a modification of the MIMC method 
which allows one to use standard Markov chain
Monte Carlo (MCMC) algorithms to replace independent and coupled sampling,
in certain contexts. We prove a variance theorem for a simplified estimator 
which shows that using our MIMCMC method is preferable, in the sense above, 
to i.i.d.~sampling from the most accurate approximation, 
under appropriate assumptions. 
The method is numerically illustrated on 
a Bayesian inverse problem associated to a stochastic partial differential equation (SPDE),
where the path measure is conditioned on some observations.}
\\
\noindent \textbf{Key words}: 
{Multi-Index Monte Carlo, Markov chain Monte Carlo, Stochastic Partial Differential equations.}


\section{Introduction}

Stochastic systems associated to 
discretization over
multiple dimensions occur in a wide range of applications. 
For instance, such stochastic
systems can represent a process that evolves in both space and time, such as
stochastic partial differential equations (SPDEs) 
and random partial differential equations. 
See for instance \cite{beskos1} for a list of applications.  
In this article, we are interested in the case where we want to compute expectations 
with respect to (w.r.t.) such probability laws. 
In most practical applications of interest, the computation of the expectations 
is not analytically possible. 
This is for at least two reasons: (1) such probability laws
are often not tractable without some discretization and 
(2) even after discretization, the expectations are not tractable and need to be approximated.
One way to deal with this issue is to sample independently from the discretized probability law, 
and use the Monte Carlo method.

One well-known method for improving over Monte
Carlo is the popular Multilevel Monte Carlo (MLMC) method \cite{giles,giles1,hein}. This approach introduces a hierarchy of discretizations, and a 
telescopic sum representation of the expectation of interest.  
Assuming the computational cost of sampling a discretized law increases 
as the approximation error falls, 
and that independent sampling of \emph{couples} (pairs) of the discretized laws is possible, 
then the required work to achieve a given level of error can be reduced by using MLMC. 
The requirement to of independent (or exact) sampling from couples
with the correct marginals is often not possible in many contexts.
This has been dealt with in several recent works, 
such as \cite{beskos,jasra,mlpf,mlpf1}.

In the scenario of this article, the discretization is in multiple dimensions. 
A more efficient version of the MLMC method can be designed in this case, 
called multi-index Monte Carlo (MIMC) \cite{mimc}. 
The method essentially relies on being able to independently sample 
from $2^d$ terms 
in a dependent manner, where $d$ is the number of dimensions which are discretized. 
We will expand upon this point later on,
but the idea is to first construct a new telescopic representation of the expectation 
with respect to the most accurate probability law, in terms of
differences of differences (for $d=2$, or differences$^d$) of expectations. 
These higher order differences are then approximated using correlated random variables. 
Again, assuming the computational cost of sampling a discretized law increases as the approximation error falls, then  
the work to achieve a given level of error is reduced by using MIMC.
Under suitable regularity conditions, and assuming a suitable choice of indices is chosen,
this can also be preferable to MLMC.

In this article we consider the case when such probability laws (or the couplings) 
are too complex to be sampled independently. 
This occurs for example when the measure of a stochastic process
is conditioned on real data, such as in real-time data assimilation \cite{law2015data} 
or online filtering \cite{llospis}, or in a static Bayesian inverse problem 
\cite{stuart2010inverse, hoang2014determining}.
In the simplest case this means that the probability
measure of the conditioned process can only be evaluated up to a normalizing constant,
but cannot be simulated from.
We develop a modification of the MIMC method which allows one to use standard MCMC algorithms to replace independent and coupled sampling, in certain contexts. 
We prove a variance theorem which shows that using
our MIMCMC method is preferable to using independent identically distributed (i.i.d.) 
random variables from the most accurate approximation, under appropriate assumptions
and in the sense of cost to obtain a given error tolerance. The proof is however, for a simplified estimator
and not the one implemented. 
The method is illustrated on a Bayesian inverse problem associated to an SPDE.

This paper is structured as follows. In Section \ref{sec:model} the exact context is given along with a short review of the MIMC method.
In Section \ref{sec:approach} our approach is outlined, along with a variance result. In Section \ref{sec:numerics} numerical results
are presented. The appendix includes a technical result used in our variance theorem.

\section{Modelling Context}
\label{sec:model}

We are interested in a random variable $x\in\mathsf{X}$, 
with $\sigma-$algebra $\mathcal{X}$, 
for which we want to compute
expectations of real-valued bounded and measurable functions $\varphi:\mathsf{X}\rightarrow\mathbb{R}$, $\mathbb{E}[\varphi(X)]$. 
We assume that the random variable $X$ is such that it is associated to a continuum system such as
a stochastic partial differential equation (SPDE). In practice one can only hope to evaluate a discretised
version of the random variable.

 Let $\alpha=\alpha_{1:d}=(\alpha_1,\dots,\alpha_d)\in\mathbb{N}_0^d$.
For any fixed and finite-valued 
index $\alpha$, one 
can obtain a biased approximation $X_{\alpha}\in\mathsf{X}_{\alpha}\subseteq\mathsf{X}$
(with $\sigma-$algebra $\mathcal{X}_{\alpha}$), where
we use the convention $\mathsf{X}_{\infty,\dots,\infty}=\mathsf{X}$.
Let $\varphi:\mathsf{X}\rightarrow\mathbb{R}$.
If $x\in\mathsf{X}$ and  $x \notin \mathsf{X}_{\alpha}$ for any $\alpha$ with $\alpha_i<\infty$ for some $i$, 
then $\varphi(x)$ is written, and if $x\in\mathsf{X}_{\alpha}$ for some $\alpha<\infty$, 
then $\varphi_{\alpha}(x)$ is used.
In general, $\mathbb{E}[\varphi_{\alpha}(X_{\alpha})]\neq\mathbb{E}[\varphi(X)]$, but
\begin{equation}
\lim_{\min_{1\leq i\leq d}\alpha_i \rightarrow+\infty}|\mathbb{E}[\varphi_{\alpha}(X_{\alpha})]-\mathbb{E}[\varphi(X)]| = 0.\label{eq:assump1}
\end{equation}
It is assumed that the computational cost associated with $X_{\alpha}$ increases as the values of $\alpha$
increase. We constrain $\alpha_{1:d}\in\mathcal{I}_{ L_1:L_d}  := \{\alpha\in\mathbb{N}_0^d:\alpha_1\in\{0,\dots,L_1\},\dots,\alpha_d\in\{0,\dots,L_d\}\}$.

To make things more precise, we assume that the probability measure of $X$ and $X_{\alpha}$
is defined as follows. Consider observations $y\in\mathsf{Y}$ and a likelihood function in $y$,
$g:\mathsf{Y}\times\mathsf{X}\rightarrow\mathbb{R}_+$.
When  $x\in\mathsf{X}$ and  $x \notin \mathsf{X}_{\alpha}$ for any $\alpha$ with $\alpha_i<\infty$ for some $i$, 
we write $g(y|x)$, and when $x\in\mathsf{X}_{\alpha}$ we write $g_{\alpha}(y|x)$. In both
situations $\int_{\mathsf{Y}}g(y|x)\dif y=\int_{\mathsf{Y}}g_{\alpha}(y|x')\dif y=1$ for any $(x,x')\in\mathsf{X}\times\mathsf{X}_{\alpha}$
and $\dif y$ a dominating measure.

We have for $x\in\mathsf{X}$ ,
$$
\pi(\dif x|y) \propto g(y|x) p(\dif x) ,
$$
with $p$ a probability measure on $\mathsf{X}$,
and for $x_{\alpha}\in\mathsf{X}_{\alpha}$ ,
$$
\pi_{\alpha}(\dif x|y) \propto g_{\alpha}(y|x) p_{\alpha}(\dif x) ,
$$
with $p_{\alpha}$ a probability measure on $\mathsf{X}_{\alpha}$.

\subsection{MIMC Methods}

Write $\mathbb{E}_{\alpha}$ as expectation w.r.t.~$\pi_{\alpha}$ 
and $\mathbb{E}$ as expectation w.r.t.~$\pi$.
Define the difference operator $\Delta_i$, $i\in\{1,\dots,d\}$ as
$$
\Delta_i \mathbb{E}_{\alpha}[\varphi_\alpha(X_{\alpha})]
= \left\{\begin{array}{ll}
\mathbb{E}_{\alpha}[\varphi_{\alpha}(X_{\alpha})]- \mathbb{E}_{\alpha-e_i}[\varphi_{\alpha-e_i}(X_{\alpha-e_i})]  & \textrm{if}~\alpha_i>0 \\
\mathbb{E}_{\alpha}[\varphi_{\alpha}(X_{\alpha})] & \textrm{o/w}
\end{array}\right.
$$
where $e_i$ are the canonical vectors on $\mathbb{R}^d$. 
Set $\Delta = \bigotimes_{i=1}^d \Delta_i := 
\Delta_d \cdots \Delta_1$.  
Observe the collapsing identity
\begin{equation}\label{eq:collapse}
\mathbb{E}[\varphi(X)] = \sum_{\alpha\in \mathbb{N}_0^d}
\Delta \mathbb{E}_{\alpha}[\varphi_{\alpha}(X_{\alpha})].
\end{equation}

Letting $\mathcal{I} \subset \mathbb{N}_0^d$, 
\cite{mimc} consider the biased approximation of $\mathbb{E}[\varphi(X)]$ 
given by 
\begin{equation}\label{eq:mimc}
\sum_{\alpha\in\mathcal{I}}\Delta \mathbb{E}_{\alpha}[\varphi_{\alpha}(X_{\alpha})].
\end{equation}
Each summand can be estimated by Monte Carlo, 
coupling the $1<k_\alpha \leq 2^d$ probability measures with indices 
$\alpha' = \alpha(1),\dots, \alpha(k_\alpha)$,
for a given $\alpha\in\mathcal{I}$.  
That is, for this given $\alpha$, 
one draws an i.i.d.~sample $(X_{\alpha(1)}^i,\dots,X_{\alpha(k_\alpha)}^i)$ 
for $i=1,\dots, N_\alpha$, such that each $X_{\alpha(k)}^i$
for $k=1,\dots, k_\alpha$ is correlated to each other one, and with the appropriate marginal.
Denote this approximation by 
$$
\Delta \mathbb{E}^{N_\alpha}_{\alpha}[\varphi_{\alpha}(X_{\alpha})] := 
\frac{1}{N_\alpha} \sum_{i=1}^{N_\alpha} (\Delta \varphi_\alpha)(X^i_{\alpha(1)},\dots, X^i_{\alpha(k_\alpha)}) \, .
$$

Following the MLMC analysis, the mean square error (MSE) 
of the MIMC estimator is decomposed as
\begin{equation}
\begin{split}
& \bbE\left [ \left (\sum_{\alpha\in\mathcal{I}} \Delta \mathbb{E}^{N_\alpha}_{\alpha}[\varphi_{\alpha}(X_{\alpha})] 
- \mathbb{E}[\varphi(X)] \right )^2 \right ] = \\
& \underbrace{\bbE\left [ \left (\sum_{\alpha\in\mathcal{I}} 
( \Delta \mathbb{E}^{N_\alpha}_{\alpha}[\varphi_{\alpha}(X_{\alpha})] 
- \Delta \mathbb{E}_{\alpha}[\varphi_{\alpha}(X_{\alpha})] ) \right )^2 \right ] }_{variance}
+ \underbrace{\left (  \sum_{\alpha\notin\mathcal{I}} 
\Delta \mathbb{E}_{\alpha}[\varphi_{\alpha}(X_{\alpha})] \right )^2}_{bias^2},
\end{split}\label{eq:mimse}
\end{equation}
where \eqref{eq:collapse} and \eqref{eq:mimc} were used.

The following assumptions are made in \cite{mimc}.
\begin{ass}[MIMC Assumptions]
\label{ass:mimc1} 
There is some $C>0$ and there are some $w_i, \beta_i, \gamma_i >0$ 
for $i=1,\dots, d$, such that the following estimates hold
\begin{itemize}
\item[{\rm (a)}] $\left|\Delta \mathbb{E}_{\alpha}[\varphi_{\alpha}(X_{\alpha})]\right| 
=: B_\alpha \leq C \prod_{i=1}^d 2^{-w_i\alpha_i}$;
\item[{\rm (b)}] $\bbE\left [ \left ( 
\Delta \mathbb{E}^{N_\alpha}_{\alpha}[\varphi_{\alpha}(X_{\alpha})] 
- \Delta \mathbb{E}_{\alpha}[\varphi_{\alpha}(X_{\alpha})]\right )^2 \right ] 
=: N_\alpha^{-1}V_\alpha \leq C N_\alpha^{-1}\prod_{i=1}^d 2^{-\beta_i\alpha_i}$;
\item[{\rm (c)}] {\rm Cost}$(X_{\alpha}) =: C_\alpha \leq C \prod_{i=1}^d 2^{\gamma_i \alpha_i}.$
\end{itemize}
\end{ass}

In the present work we will constrain our attention to
\begin{equation}\label{eq:tp}
\mathcal{I}_{ L_1:L_d}  := \{\alpha\in\mathbb{N}_0^d:\alpha_1\in\{0,\dots,L_1\},\dots,\alpha_d\in\{0,\dots,L_d\}\} \ . 
\end{equation}
\noindent
Define $\mathcal{A}(\alpha^*) = \{\alpha\in  \mathbb{N}_0^d ; \alpha_j \geq \alpha^*_j,~ {\rm for~at~least~one}~ j =1,\dots, d\}$.
Observe that 
\begin{equation}\label{eq:biassum}
\sum_{\alpha\notin\mathcal{I}_{ L_1:L_d}} 
\Delta \mathbb{E}_{\alpha}[\varphi_{\alpha}(X_{\alpha})]
\lesssim 
\sum_{\alpha \in \mathcal{A}(L_1:L_d)} \prod_{i=1}^d 2^{-w_i\alpha_i}
\le  \sum_{i=1}^d 2^{-w_iL_i}. 
\end{equation}
This is also consistent with a triangle-inequality estimate
of the bias from 
$$
\mathbb{E}_{(L_1:L_d)}[\varphi_{(L_1:L_d)}(X_{(L_1:L_d)})] 
= \sum_{\alpha\in\mathcal{I}_{ L_1:L_d}} 
\Delta \mathbb{E}_{\alpha}[\varphi_{\alpha}(X_{\alpha})] \ ,
$$
under the reasonable assumption that Assumption \ref{ass:mimc1}
arises from individual estimates of the form
$\Delta_i \mathbb{E}_{\alpha}[\varphi_{\alpha}(X_{\alpha})] = \cO(2^{-w_i\alpha_i})$,
coupled with mixed regularity conditions.

Now, suppose we aim to satisfy an MSE bound of $\cO(\varepsilon^2)$.

\begin{prop}[MIMC cost]\label{pro:mimc}
Given Assumption \ref{ass:mimc1}, with $\beta_i>\gamma_i$, for all $i=1,\dots, d$,
and assuming 
$\sum_{i=1}^d \gamma_i/w_i 
\leq 2$,
it is possible to identify $(L_1,\dots,L_d)$
and $\{ N_\alpha \}_{\alpha \in \cI_{ L_1:L_d}}$ such that for $C>0$
$$
\bbE\left [ \left (\sum_{\alpha\in\mathcal{I}_{ L_1:L_d}} \Delta \mathbb{E}^{N_\alpha}_{\alpha}[\varphi_{\alpha}(X_{\alpha})] 
- \mathbb{E}[\varphi(X)] \right )^2 \right ] \leq C \varepsilon^2,
$$
for a cost of $\cO(\varepsilon^{-2})$.
\end{prop}
\begin{proof}
Following from \eqref{eq:biassum}, the condition $L_j=|\log(\varepsilon/d)|/w_j$,
for all $j=1,\dots, d$, is sufficient to control the bias term in \eqref{eq:mimse}.
Given $\cI$, in this case constrained to be of the form $\mathcal{I}_{ L_1:L_d}$,
the $N_\alpha$ are optimized in the same way as MLMC so that 
$N_\alpha = \lceil(\varepsilon^{-2} K_\cI (V_\alpha/C_\alpha)^{1/2}\rceil$,
 where $K_\cI = \sum_{\alpha\in \cI} (V_\alpha C_\alpha)^{1/2}$ and $\lceil \cdot \rceil$
 denotes the integer ceiling of a non-integer, ensuring $N_\alpha\geq1$.
The cost is $\cO(\varepsilon^{-2}K_\cI^2)$.
See \cite{giles, ours, mlpf} for details.
For $\mathcal{I}_{ L_1:L_d}$ from \eqref{eq:tp} we have 
\begin{equation}\label{eq:tpcost}
K_\cI = \sum_{\alpha\in \cI} (V_\alpha C_\alpha)^{1/2} \le C\sum_{\alpha\in\cI}\prod_{i=1}^d 2^{\alpha_i(\gamma_i-\beta_i)/2}= 
C\prod_{i=1}^d \sum_{\alpha_i=1}^{L_i} 2^{\alpha_i(\gamma_i-\beta_i)/2}.
\end{equation}

Notice that $C_{(L_1:L_d)} \propto \varepsilon^{-\sum_{i=1}^d \gamma_i/w_i}$.
The constraint that 
$\sum_{i=1}^d \gamma_i/w_i 
\leq 2$ 
ensures that 
$C_{(L_1:L_d)} \lesssim \varepsilon^{-2}$, 
so the cost is dominated by $\varepsilon^{-2}$,
even if the theoretically optimal $N_{(L_1:L_d)}$
falls below 1, so that $N_{(L_1:L_d)}=1$.
\end{proof}

Notice that as usual the asymptotic relationship Cost$(\varepsilon)$ 
is determined by the signs of $\gamma_i-\beta_i$ for $i=1,\dots,d$.  
The proposition above shows that if $\beta_i>\gamma_i$ for all $i$, 
then one obtains the optimal {\it dimension-independent} 
cost of $\cO(\varepsilon^{-2})$.  
The other cases follow similarly from the relationship \eqref{eq:tpcost}.
The general case is considered in \cite{mimc}.  
If $\sum_{i=1}^d \gamma_i/w_i > 2$ 
then the theoretically optimal $N_{(L_1:L_d)}<1$,
and furthermore when we set $N_{(L_1:L_d)}=1$
then the cost will be dominated by $C_{(L_1:L_d)}$.

\begin{rem}[Choice of index set]
It is shown in \cite{mimc} that in fact it can be preferable 
to consider more complex index sets $\mathcal{I}$ 
than the tensor product one considered here, such
as $\cI_{\delta,L} =
\{\alpha \in  \mathbb{N}_0^d ; \alpha \cdot \delta \leq L, \delta \in (0,1]^d, \sum_{i=1}^d \delta_i =1 \}$. 
Indeed for any convex set $\cI \subseteq \cI_{ L_1:L_d} \subset \mathbb{N}_0^d$
other than $\cI_{ L_1:L_d}$, the bias will be larger, 
including more terms associated to the missing terms in the collapsing sum approximation
of $\mathbb{E}_{(L_1:L_d)}[\varphi_{(L_1:L_d)}(X_{(L_1:L_d)})]$.
However, each term left out saves a certain cost.  
Convexity ensures more expensive and smaller bias terms are excluded.
Since the present work is concerned with proof of principle,
this enhancement is left to future work.
\end{rem}

\section{Approach}
\label{sec:approach}

We consider \eqref{eq:mimc} and a given summand for $\alpha\in\mathcal{I}_{L_1:L_d}$. 
We suppose that there are
$1<k_{\alpha}\leq 2^d$ probability measures for which one wants to compute an expectation 
(in the case that there is only 1, one can use an ordinary Monte Carlo/ MCMC method 
to compute the expectation). 
These $k_{\alpha}$ probability measures induce $k_{\alpha}'=k_\alpha/2$ 
differences in \eqref{eq:mimc}.  
Our approach will estimate each summand of \eqref{eq:mimc}  independently.

For simplicity of notation we will write the associated random variables and indices $X_{\alpha(1)},\dots,X_{\alpha(k_{\alpha})}$.
The convention of the labelling is such that, writing $\alpha(i)_j$ as the $j^{th}-$element of $\alpha(i)$, 
$\sum_{j=1}^d[\alpha(2i)-\alpha(2i-1)]_j=1$ for each $i\in\{1,\dots,k_{\alpha}'\}$, and 
$\sum_{j=1}^d[\alpha(i)-\alpha(i-1)]_j2^{j-1}\geq 0$ for each $i\in\{2,\dots,k_{\alpha}\}$.
That is, $X_{\alpha(k_{\alpha})}$ (=$X_{\alpha}$)
is the most expensive random variable and $X_{\alpha(1)}$ (if $k_{\alpha}=2^d$, it is $X_{\alpha-\sum_{j=1}^d e_j}$) the cheapest random variable.
We suppose that it is possible to construct a dependent coupling of the prior 
$Q_\alpha$ on 
$\bigotimes_{k=1}^{k_{\alpha}} \mathsf{X}_{\alpha(k)} := \mathsf{X}_{\alpha(1)} \times \dots \times \mathsf{X}_{\alpha(k_\alpha)}$,
i.e.~that for $A_i\in\mathcal{X}_{\alpha(i)}$ and $i\in\{1,\dots,k_{\alpha}\}$
$$
\int_{\mathsf{X}_{\alpha(1)}\times\cdots\times\mathsf{X}_{\alpha(i-1)}\times A_i\times 
\mathsf{X}_{\alpha(i+1)}\times\cdots\times\mathsf{X}_{\alpha(k_{\alpha})}} Q_\alpha(\dif (x_{\alpha(1)},\dots,x_{\alpha({k_{\alpha}})})) = 
p_{\alpha(i)}(A_i).
$$
Expectations and variances w.r.t.~$Q_\alpha$ are written $\mathbb{E}_{Q_\alpha}$ and 
$\mathbb{V}\textrm{ar}_{Q_\alpha}$.
This is possible in some SPDE contexts (e.g.~\cite{spde_disc}). 
Let $G:\mathbb{N}_0^d\times \bigotimes_{k=1}^{k_{\alpha}} \mathsf{X}_{\alpha(k)}\rightarrow (0,\infty)$. 
We propose to sample from the {\it approximate coupling} 
$$
\Pi_{\alpha}(\dif (x_{\alpha(1)},\dots,x_{\alpha(k_{\alpha})})) \propto G_{\alpha}(x_{\alpha(1)},\dots,x_{\alpha(k_{\alpha})}) Q_\alpha(\dif (x_{\alpha(1)},\dots,x_{\alpha(k_{\alpha})})).
$$
Expectations w.r.t.~this probability measure are written $\mathbb{E}_{\Pi_{\alpha}}$.
One sensible choice of $G_{\alpha}(x_{\alpha(1)},\dots,x_{\alpha(k_{\alpha})})$, 
and the one which is assumed henceforth, is
$$
G_{\alpha}(x_{\alpha(1)},\dots,x_{\alpha(k_{\alpha})}) = \max\{g_{\alpha(1)}(y|x_{\alpha(1)}),\dots, g_{\alpha(k_{\alpha})}(y|x_{\alpha(k_{\alpha})})\}.
$$
This ensures that the variance of the approach to be introduced is upper-bounded by a finite constant.
Then for any $\alpha(i)$, $i\in\{1,\dots,k_{\alpha}\}$,
\begin{eqnarray}
\mathbb{E}_{\alpha(i)}[\varphi_{\alpha(i)}(X_{\alpha(i)})] & = &
\mathbb{E}_{\Pi_{\alpha}}\Big[\varphi_{\alpha(i)}(X_{\alpha(i)})\frac{g_{\alpha(i)}(y|X_{\alpha(i)})}{G_{\alpha}(X_{\alpha(1)},\dots,X_{\alpha(k_{\alpha})}) }\Big]\Bigg/ \nonumber\\ 
& &
\mathbb{E}_{\Pi_{\alpha}}\Big[
\frac{g_{\alpha(i)}(y|X_{\alpha(i)})}{G_{\alpha}(X_{\alpha(1)},\dots,X_{\alpha(k_{\alpha})}) }\Big] \, .
\label{eq:main_id}
\end{eqnarray}
To ease the subsequent notations, set for any $\alpha(i)$, $i\in\{1,\dots,k_{\alpha}\}$, 
\begin{equation}\label{eq:h}
H_{i,\alpha}(x_{\alpha(1)},\dots,x_{\alpha(k_{\alpha})}) = \frac{g_{\alpha(i)}(y|x_{\alpha(i)})}{G_{\alpha}(x_{\alpha(1)},\dots,x_{\alpha(k_{\alpha})}) } \, .
\end{equation}

\subsection{Method and Analysis}

Let $k_{\alpha}=2^d$, and $k_{\alpha}'=2^{d-1}$.
Now, to approximate the summand in \eqref{eq:mimc}, we have
$$
\Delta \mathbb{E}_{\alpha}[\varphi_{\alpha}(X_{\alpha})] = \sum_{i=1}^{k_{\alpha}'}
(-1)^{|\alpha(k_{\alpha})-\alpha(2i)|}
\{\mathbb{E}_{\alpha(2i)}[\varphi_{\alpha(2i)}(X_{\alpha(2i)})] - 
\mathbb{E}_{\alpha(2i-1)}[\varphi_{\alpha(2i-1)}(X_{\alpha(2i-1)})]
\} \, . 
$$
Then, we have that via \eqref{eq:main_id}
\begin{eqnarray*}
\Delta \mathbb{E}_{\alpha}[\varphi_{\alpha}(X_{\alpha})] & = &
\sum_{i=1}^{k_{\alpha}'}(-1)^{|\alpha(k_{\alpha})-\alpha(2i)|}
\Bigg\{
\frac{
\mathbb{E}_{\Pi_{\alpha}}[\varphi_{\alpha(2i)}(X_{\alpha(2i)})
H_{2i,\alpha}(X_{\alpha(1)},\dots,X_{\alpha(k_{\alpha})})]
}
{
\mathbb{E}_{\Pi_{\alpha}}[H_{2i,\alpha}(X_{\alpha(1)},\dots,X_{\alpha(k_{\alpha})})]
}
- \\ & & 
\frac{
\mathbb{E}_{\Pi_{\alpha}}[\varphi_{\alpha(2i-1)}(X_{\alpha(2i-1)})
H_{2i-1,\alpha}(X_{\alpha(1)},\dots,X_{\alpha(k_{\alpha})})]
}
{
\mathbb{E}_{\Pi_{\alpha}}[H_{2i-1,\alpha}(X_{\alpha(1)},\dots,X_{\alpha(k_{\alpha})})]
}
\Bigg\} \, ,
\end{eqnarray*}
where we recall that $H_{i,\alpha}$ is defined in \eqref{eq:h}.

This identity can be approximated by running an ergodic $\Pi_{\alpha}-$invariant Markov kernel $\mathbb{K}_\alpha$ on the 
space $(\mathsf{Z} = \bigotimes_{k=1}^{k_{\alpha}} \mathsf{X}_{\alpha(k)}$, $\mathcal{Z} = \bigvee_{k=1}^{k_{\alpha}} \mathcal{X}_{\alpha(k)})$.
Write the Markov chain run for $N-$steps as $\{X^j_{\alpha(1)},\dots,X^j_{\alpha(k_{\alpha})}\}_{1\leq j \leq N}$.
Then the approximation of $\Delta \mathbb{E}_{\alpha}[\varphi_{\alpha}(X_{\alpha})]$ is
\begin{eqnarray}\label{eq:mimcmc_full}
\sum_{i=1}^{k_{\alpha}'}(-1)^{|\alpha(k_{\alpha})-\alpha(2i)|}
\Bigg\{
\frac{
\frac{1}{N}\sum_{j=1}^N
\varphi_{\alpha(2i)}(x^j_{\alpha(2i)})
H_{2i,\alpha}(x^j_{\alpha(1)},\dots,x^j_{\alpha(k_{\alpha})})
}
{
\frac{1}{N}\sum_{j=1}^N
H_{2i,\alpha}(x^j_{\alpha(1)},\dots,x^j_{\alpha(k_{\alpha})})
} -
\\
\nonumber
\frac{
\frac{1}{N}\sum_{j=1}^N
\varphi_{\alpha(2i-1)}(x^j_{\alpha(2i-1)})
H_{2i-1,\alpha}(x^j_{\alpha(1)},\dots,x^j_{\alpha(k_{\alpha})})
}
{
\frac{1}{N}\sum_{j=1}^N
H_{2i-1,\alpha}(x^j_{\alpha(1)},\dots,x^j_{\alpha(k_{\alpha})})
} 
\Bigg\} \, .
\end{eqnarray}

We now give a result on the variance of this approach. 
However, this is for the simplified estimator
\begin{equation}\label{eq:vp_def}
\widehat{\varphi}^N_\alpha := \sum_{i=1}^{k_{\alpha}'}(-1)^{|\alpha(k_{\alpha})-\alpha(2i)|}
\Bigg\{
\frac{
\frac{1}{N}\sum_{j=1}^N
\varphi_{\alpha(2i)}(x^j_{\alpha(2i)})
H_{2i,\alpha}(x^j_{\alpha(1)},\dots,x^j_{\alpha(k_{\alpha})})
}
{
\mathbb{E}_{\Pi_{\alpha}}[H_{2i,\alpha}(X_{\alpha(1)},\dots,X_{\alpha(k_{\alpha})})]
} -
\end{equation}
$$
\frac{
\frac{1}{N}\sum_{j=1}^N
\varphi_{\alpha(2i-1)}(x^j_{\alpha(2i-1)})
H_{2i-1,\alpha}(x^j_{\alpha(1)},\dots,x^j_{\alpha(k_{\alpha})})
}
{
\mathbb{E}_{\Pi_{\alpha}}[H_{2i-1,\alpha}(X_{\alpha(1)},\dots,X_{\alpha(k_{\alpha})})]
} 
\Bigg\}.
$$
The analysis of this estimator is non-trivial, but significantly more straightforward than the one implemented, which
is left for future work. We believe the same result to hold for the estimate used in practice \eqref{eq:mimcmc_full}. 
The challenge for the implemented estimator \eqref{eq:mimcmc_full} is associated to treating 
differences of differences for self-normalized estimators, which
does not appear to exist yet in the literature.
A bounded function on $\mathsf{Z}$ means a function that is uniformly upper bounded, i.e.~the upper-bound does not depend on $\alpha$ (although it may depend on $d$).

\begin{ass}[MIMCMC Assumptions]\nonumber
We assume the following:
\begin{hypA}
\label{hyp:A}
For every $y\in\mathsf{Y}$
there exist $0< \underline{C}< \overline{C}<+\infty$
such that for every $\alpha(i)$, $i\in\{1,\dots,k_{\alpha}\}$, 
$x\in\mathsf{X}_{\alpha(i)}$, 
$$
\underline{C} \leq g_{\alpha(i)}(y|x) \leq \overline{C}.
$$
\end{hypA}
\begin{hypA}
\label{hyp:B}
For $\varphi:\mathsf{X}\rightarrow\mathbb{R}$ bounded,
and $f:\mathbb{N}_0^d\times \bigotimes_{k=1}^{k_{\alpha}} \mathsf{X}_{\alpha(k)}\rightarrow \mathbb{R}$ 
bounded, we have
$$
\left|\mathbb{E}_{Q_\alpha}\Big[
f_{\alpha}(X_{\alpha(1)},\dots,X_{\alpha(k_{\alpha})})\Big\{
\sum_{i=1}^{k_{\alpha}'}(-1)^{|\alpha(k_{\alpha})-\alpha(2i)|}\{
\varphi_{\alpha(2i)}(X_{\alpha(2i)}) - 
\varphi_{\alpha(2i-1)}(X_{\alpha(2i-1)})\}\Big\}
\Big]\right|
$$
$$
\leq C_1(\alpha)
$$
with $\lim_{\min_{1\leq i\leq d}\alpha_i \rightarrow+\infty}C_1(\alpha)=0$.
For $\varphi:\mathsf{X}\rightarrow\mathbb{R}$ bounded, we have
$$
\mathbb{V}\textrm{ar}_{Q_\alpha}\Big[
\sum_{i=1}^{k_{\alpha}'}(-1)^{|\alpha(k_{\alpha})-\alpha(2i)|}\{
\varphi_{\alpha(2i)}(X_{\alpha(2i)}) - 
\varphi_{\alpha(2i-1)}(X_{\alpha(2i-1)})\}
\Big] \leq 
C_2(\alpha)^2
$$
with $\lim_{\min_{1\leq i\leq d}\alpha_i \rightarrow+\infty}C_2(\alpha)=0$.
\end{hypA}
\begin{hypA}
\label{hyp:C}
There exist a $\xi\in(0,1)$ and a probability measure $\nu_\alpha$ on $(\mathsf{Z},\mathcal{Z})$ for every $\alpha$ such that 
$$
\mathbb{K}_\alpha(z,A)  \geq \xi \nu_\alpha(A)\ (z\in \mathsf{Z}, A\in\mathcal{Z}). 
$$
$\mathbb{K}_\alpha$ is $\Pi_{\alpha}-$reversible. 
\end{hypA}
\end{ass}

Let $D(\alpha) = 
\max\{
C_1^2,C_2^2,C_1C_2\},$ with $C_i(\alpha)$ 
given as above for $i=1,2$.
Set $\mathbb{E}$ as the expectation w.r.t.~the law of the simulated Markov chain. 
We have the following result.


\begin{prop}[Main result]
\label{prop:main_res}
Assume (A\ref{hyp:A}-\ref{hyp:C}). Then 
there exist a $C<+\infty$ independent of $\alpha$ such that
$$
\mathbb{E}\Bigg[\Bigg(
\sum_{i=1}^{k_{\alpha}'}(-1)^{|\alpha(k_{\alpha})-\alpha(2i)|}
\Bigg\{
\frac{
\frac{1}{N}\sum_{j=1}^N
\varphi_{\alpha(2i)}(x^j_{\alpha(2i)})
H_{2i,\alpha}(x^j_{\alpha(1)},\dots,x^j_{\alpha(k_{\alpha})})
}
{
\mathbb{E}_{\Pi_{\alpha}}[H_{2i,\alpha}(X_{\alpha(1)},\dots,X_{\alpha(k_{\alpha})})]
} -
$$
$$
\frac{
\frac{1}{N}\sum_{j=1}^N
\varphi_{\alpha(2i-1)}(x^j_{\alpha(2i-1)})
H_{2i-1,\alpha}(x^j_{\alpha(1)},\dots,x^j_{\alpha(k_{\alpha})})
}
{
\mathbb{E}_{\Pi_{\alpha}}[H_{2i-1,\alpha}(X_{\alpha(1)},\dots,X_{\alpha(k_{\alpha})})]
} 
\Bigg\}
- \Delta \mathbb{E}_{\alpha}[\varphi_{\alpha}(X_{\alpha})]\Bigg)^2
\Big]
%
\leq 
\frac{C D(\alpha)}{N}.
$$
\end{prop}
\begin{proof}
The proof is essentially the same as that of \cite[Theorem 3.1]{jasra}, 
with the exception that Proposition A.1 needs to be augmented, which is done in the appendix. 
\end{proof}

\subsubsection{MIMC considerations}

Recall \eqref{eq:vp_def} and set
\begin{equation}\label{eq:mimcmcest}
\widehat{\varphi}_{\cI_{ L_1:L_d}}^{\rm MI} := \sum_{\alpha\in \cI_{ L_1:L_d}} \widehat{\varphi}^{N_\alpha}_\alpha .
\end{equation}

Consider the following assumptions
\begin{ass}[MIMCMC rates]\label{ass:mimc2} 
There is some $C>0$ and there are some $w_i, \beta_i, \gamma_i >0$ 
for $i=1,\dots, d$, such that the following estimates hold
\begin{itemize}
\item[{\rm (a)}] $\left|\Delta \mathbb{E}_{\alpha}[\varphi_{\alpha}(X_{\alpha})] \right|
\leq C \prod_{i=1}^d 2^{-w_i\alpha_i}$;
\item[{\rm (b)}] $D(\alpha(1),\dots,\alpha(k_\alpha)) \leq C \prod_{i=1}^d 2^{-\beta_i\alpha_i}$;
\item[{\rm (c)}] {\rm Cost}$(X_{\alpha}) \leq C \prod_{i=1}^d 2^{\gamma_i \alpha_i},$
\end{itemize}
where we recall $D(\alpha(1),\dots,\alpha(k_\alpha))$ appears in Proposition 
\ref{prop:main_res} and is defined above that.
\end{ass}



\begin{prop}[MIMCMC cost]\label{pro:mimcmc}
Given Assumption \ref{ass:mimc2}, with $\beta_i>\gamma_i$, for all $i=1,\dots, d$,
and assuming 
$\sum_{i=1}^d \gamma_i/w_i
\leq 2$,
it is possible to identify $(L_1,\dots,L_d)$
and $\{ N_\alpha \}_{\alpha \in \cI_{ L_1:L_d}}$ such that
$$
\bbE\left [ \left (\widehat{\varphi}_{\cI_{ L_1:L_d}}^{\rm MI} - \mathbb{E}[\varphi(X)] \right )^2 \right ] \leq C \varepsilon^2 \ ,
$$
for some $C>0$ and for a cost of $\cO(\varepsilon^{-2})$. 
\end{prop}
\begin{proof}
Under the assumptions above, and following from Proposition \ref{prop:main_res},
the result follows in the same manner as Proposition \ref{pro:mimc}. 
\end{proof}

\begin{rem}[MLMCMC]
It is noted that in the case of a single discretized dimension 
the method presented constitutes a new 
Multilevel Markov chain Monte Carlo (MLMCMC) 
method,
which generalizes \cite{jasra}.  
Furthermore, in this case the proof of Proposition \ref{prop:main_res} 
goes through for the general estimator \eqref{eq:mimcmc_full}
rather than the simplified one with known normalization constants \eqref{eq:vp_def}.
There exist 2 other general MLMCMC methods in the literature.
The first \cite{hoang2013complexity} uses importance sampling to approximate 
the increments.
The second \cite{dodwell2015hierarchical} uses correlated MCMC kernels 
to couple the joint measures arising in the increments.  
The interesting question of which of these is the most efficient in a given 
circumstance is beyond the scope of the present work and is left to a future 
investigation.
\end{rem}

\section{Numerical Results for an SPDE}

\label{sec:numerics}

We consider as an example 
a linear SPDE with space-time white-noise forcing.

Consider the semi-linear stochastic heat equation with additive space-time white
noise on the one-dimensional domain $[0, 1]$ over the time interval $[0, T]$
with $T = 1$, i.e.,
\begin{equation}
  \frac{\partial u}{\partial t} =
  \frac{\partial^2 u}{\partial x^2} + \theta u + \sigma\dot{W}_t, 
\end{equation}
with the Dirichlet boundary condition and the initial value 
$u(x, 0) = u_0(x) = \sum_{k=1}^\infty u_{k,0} e_k(x)$,
for $x \in (0, 1)$.
Here $\dot{W}_t$ is space-time white noise, 
i.e. the time derivative of a cylindrical Brownian motion with identity covariance operator in space, 
$W_t = \sum_{k=1}^\infty w_{k,t} e_k(x)$, with $w_{k,t}$ i.i.d. scalar Brownian motions for each $k$, 
and $e_k(x)=\sqrt{2}\sin(k\pi x)$. 
In particular, the initial data is fixed as $u_{0,k}=1$ for all $k=1,\dots, K_{\max}$. 
This is a convenient example because the solution is given by an independent
collection of SDE for $k\in \bbN$, i.e.
$$
\dot{u}_k = (- \pi^2 k^2 + \theta ) u_k + \sigma \dot{w}_{k,t}.
$$
These SDE are analytically tractable, in as much as they are Gaussian.
In other words, the solution at time $t$ is given by 
$$
u_k(t) = e^{(-\pi^2 k^2 + \theta)t}u_{k,0} + N\left ( 0, \frac{\sigma^2(1-e^{2(\theta-\pi^2k^2)t})}{2(\pi^2k^2-\theta)}\right ),
$$ 
where the second term follows from Ito isometry.  
This will be useful as a benchmark for evaluating the mean square error of the 
approximations.

Pointwise observations of the process are obtained at times $t_j=j/T$ for $j=1,\dots,m$,
at $x = 1/3$ and $x = 2/3$.  
Since 
$u(x,t)=\sum_{k=1}^\infty u_k(t) e_k(x)$,
this ensures that the posterior distribution is nontrivial, 
in the sense that 
the observations involve all modes $\{u_k(t)\}_{k=1}^{K_{\max}}$ of the solution.
An additive Gaussian
observational noise with zero mean and variance $\tau^2 = 0.1$ is assumed.
The parameters 
are chosen as $\theta = 1/2$ and $\sigma = 1$. 

Define $\mathcal{G}(u) = \{[u(1/3,t_j),u(2/3,t_j)]^\top\}_{j=1}^{m}$,
such that the observations take the form $y\sim N(\mathcal{G}(u),\tau^2 I)$. 
Let $g(y|u) \propto \exp(-\frac1{2\tau^2}|y-\mathcal{G}(u)|^2)$.
The posterior is given by
\begin{equation}\label{eq:posterior}
\pi_\alpha(\dif u) \propto g(y| u) p_\alpha(\dif u),
\end{equation}
where the prior corresponds 
to the path measure of the SPDE above for $\alpha = \infty$, 
or its approximation at level $\alpha$, for a given set of parameters. 
{ The quantity of interest will be given by 
$\varphi(u) = \sum_{k=1}^\infty k^{-1} u_k(T) e_k(1/2)$.}

The exponential Euler scheme in \cite{spde_disc} 
will be used for 
discretization.  In other words, 
for a $K_\alpha$-mode approximation with time-resolution
$h_t=T/M_\alpha$, the solution, for $n=0,1,\dots,M_\alpha-1$, is given by
$$
u_{\alpha,k,n+1} = e^{-\pi^2k^2h_t}u_{\alpha,k,n} + 
\frac{1-e^{-\pi^2 k^2 h_t}}{\pi^2 k^2} \theta  u_{\alpha,k,n} +
\xi_{k,n}, \quad \xi_{k,n}\sim N\left ( 0, \frac{\sigma^2(1-e^{-2\pi^2k^2h_t})}{2\pi^2k^2}\right ).
$$
The quantity of interest for a multi-index $\alpha=(\alpha_x, \alpha_t)$ is given by 
$$
\varphi_\alpha(u_\alpha) =  \sum_{k=1}^{K_\alpha} { k^{-1}} u_{\alpha,k,M_\alpha} e_k(1/2). 
$$
For a given $\alpha\in \mathbb{N}^2$, we take 
$K_\alpha=K_0\times2^{\alpha_x}$ and 
$M_\alpha=M_0 \times 2^{\alpha_t}$.
In order to approximate $\Delta \varphi_\alpha(u_\alpha)$, 
we begin with an approximation of the highest resolution system $u_\alpha$.
For approximations involving $\alpha_x-1$, 
we retain only the subset of the first $K_{\alpha-e_x}$ modes.
{ For approximations involving $\alpha_t-1$,
we replace $\xi_{k,n}$ with 
$\hat \xi_{k,n} = e^{-\pi^2k^2T/M_\alpha}\xi_{k,2n} + \xi_{k,2n+1}$, 
for $n=0,1,\dots, M_{\alpha-e_t}-1$.  
This appropriate coupling is derived in Section 4.3 of \cite{chernov2016multilevel}. } 

\begin{figure}
  \includegraphics{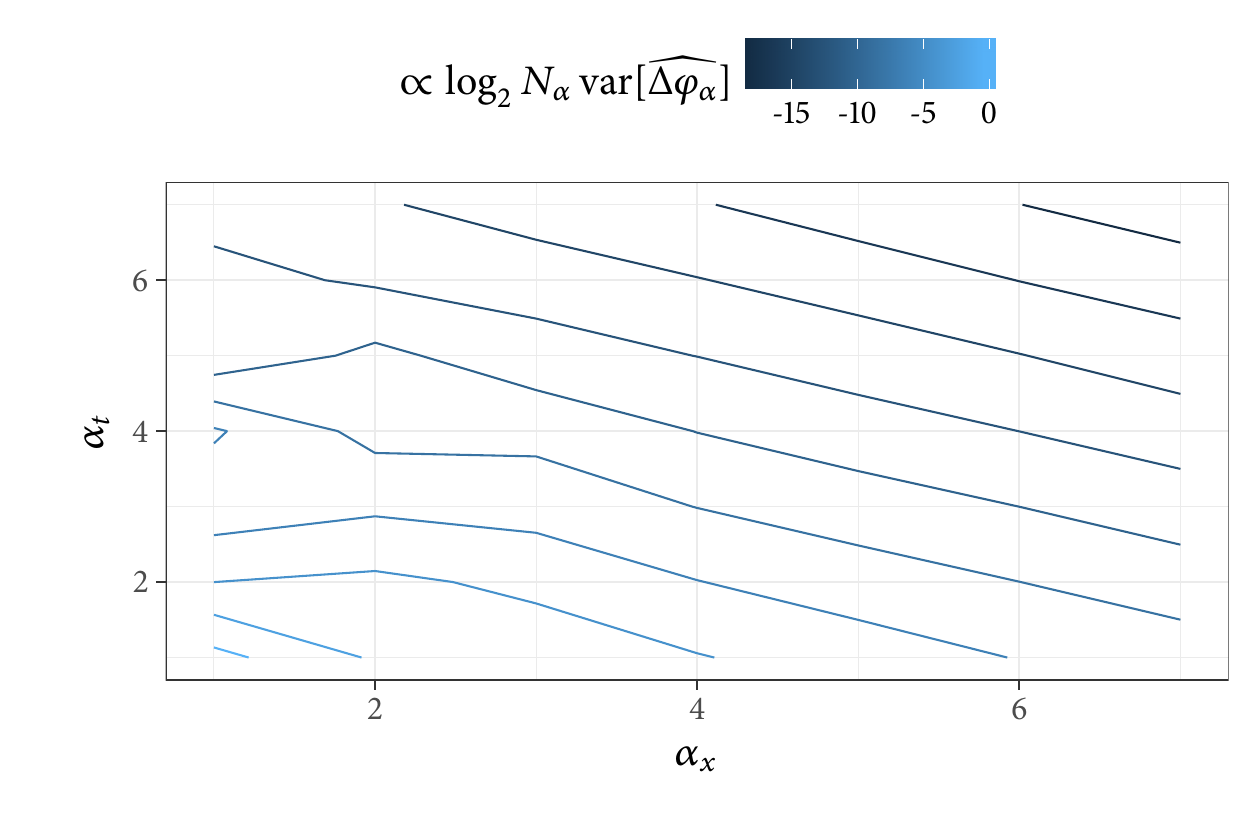}
  \caption{Estimated variance of the multi-increments over a grid of multi-indices.}
  \label{fig:rates}
\end{figure}

Note that if we can generate a proposal kernel for Metropolis-Hastings 
which keeps 
$P_\alpha(\dif u_\alpha)$ invariant, then we can use this to generate a coupled 
proposal kernel which keeps $Q_{\alpha}$ invariant.
The target is continuous with respect to $Q_{\alpha}$,
so this is sufficient for (A\ref{hyp:C}), given (A\ref{hyp:A}).
More specifically, notice that $P_\alpha$ is generated by a $K_\alpha M_\alpha$
dimensional standard Gaussian $N(0,I)$.
We keep this measure invariant by using the following pCN proposal 
\cite{beskos1} within Metropolis-Hastings, 
for some $\rho\in (0,1)$ to be tuned for an appropriate acceptance probability around $1/2$,
$$
X' = (1-\rho)^\frac12 X^{(n)} + \rho^\frac12 \eta_n, \quad \eta_n\sim N(0,I) \ .
$$
For each given random variable $X^{(n)}$, drawn from the pCN
proposal which keeps $N(0,I)$ invariant, we simply construct the draw 
$(u_{\alpha(1)},\dots,u_{\alpha(k_\alpha)})^{(n)}$ as described 
above, and clearly these pushed forward random variables 
will keep $Q_{\alpha}$ invariant.  
The acceptance probability will therefore depend only upon the ratios
$$
G_\alpha((u_{\alpha(1)},\dots, u_{\alpha(k_\alpha)})')/
G_\alpha((u_{\alpha(1)},\dots, u_{\alpha(k_\alpha)})^{(n)}) \, .$$

Denoting the approximate solution at time $t_n=n h_t$
by $u_{\alpha,n}$, \cite{spde_disc} provides the following estimate, for any $\epsilon>0$,
\begin{equation}\label{eq:errorest}
\sup_{n=1,\dots, M_\alpha} \left( \bbE |u(t_n) - u_{\alpha,n}|^2 \right)^{1/2} 
\leq C (K_\alpha^{-1/2 +\epsilon} + M_\alpha^{-1}\log(M_\alpha)).
\end{equation}
We postulate that the mixed regularity is sufficient for the convergence rate 
$$
\left( \bbE |\Delta \varphi_\alpha|^2 \right)^{1/2} \leq 
C 2^{-\alpha_x/2 - \alpha_t + \epsilon}.
$$
Indeed this is verified numerically, as illustrated in Figure \ref{fig:rates}.

\begin{figure}
  \includegraphics{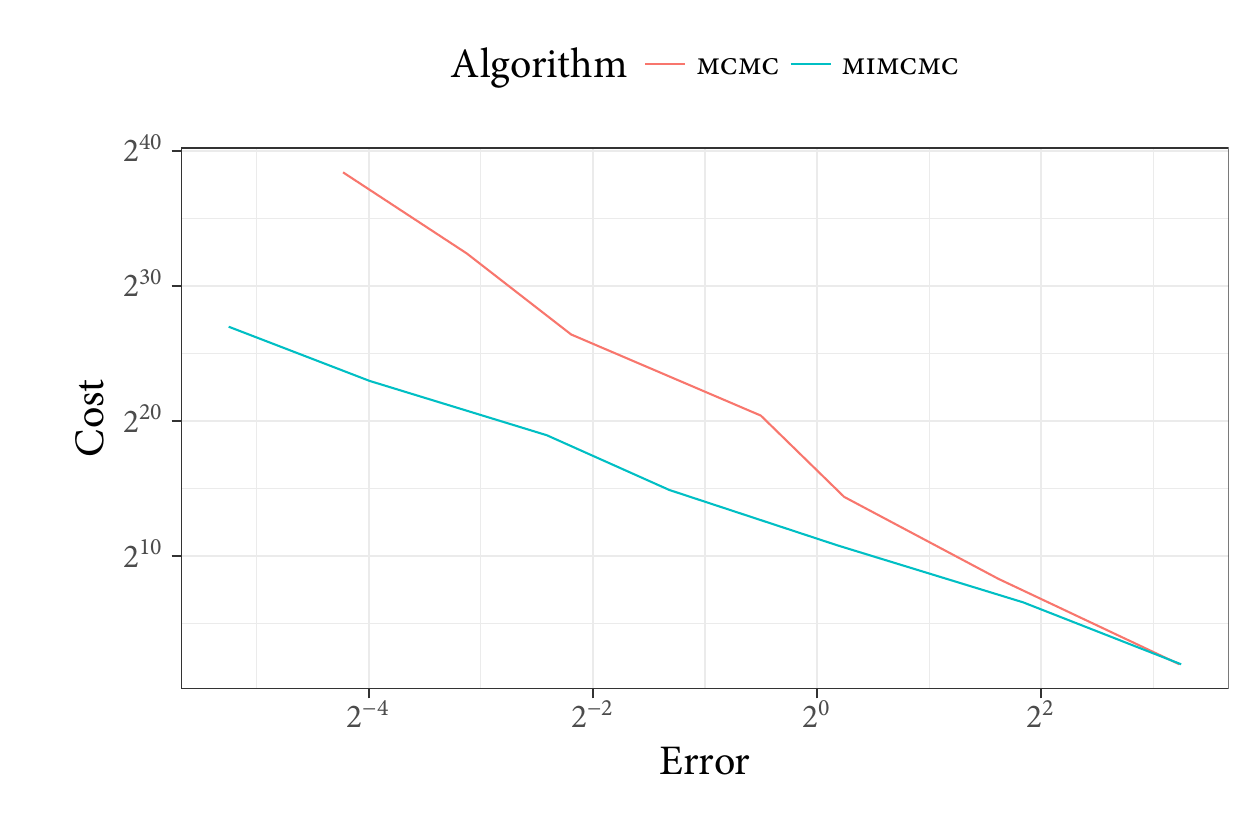}
  \caption{Cost vs error at precision levels $\alpha = (2, 1), (4, 2), \dots,
  (14, 7)$.}
  \label{fig:cost}
\end{figure}

The optimal choice of discretization according to \cite{spde_disc} is $K = M^2$, 
following from \eqref{eq:errorest} and the fact that the cost for a single realization
is proportional to $KM$.  
The main result of \cite{spde_disc} is the estimate \eqref{eq:errorest},
which 
provides a bound on the strong error proportional to $h_t=M^{-1}$, 
with a cost proportional to $M^{3}$,
for this choice of $K$.  
This provides a total cost rate for MC (or an optimal cost for MCMC) of
Cost$(\varepsilon) \propto \varepsilon^{-5}$.  
For MIMCMC, Proposition \ref{pro:mimcmc} shows that if one chooses
$L_x = 2 L_h \geq 2 |\log(\varepsilon/2)|$, 
and $N_\alpha=\varepsilon^{-2} L_x 2^{-\alpha_x-3 \alpha_t/2}$, 
then Cost$(\varepsilon) \propto \varepsilon^{-2} \log^2(\varepsilon/2)$, 
with a logarithmic penalty due to the fact that $\beta_x=\gamma_x=1$.
However in this case $\sum_{i=1}^d \gamma_i/w_i =3$,
and theoretically optimal $N_{(L_1:L_d)} =o(1)$.
When we replace this 
by $N_{(L_1:L_d)} =1$,
then the cost is dominated by $C_{(L_1:L_d)} \propto \varepsilon^{-3}$.

The true solution is computed as described in the appendix \ref{app:spde}, 
for the reference, and the MSE for is computed by comparing this to 
the results of 30 MIMCMC estimators, using the pCN method above to
 generate the driving Gaussian for each $\alpha$.
For MCMC the fitted rate is about $-5$.
For MIMCMC the fitted rate is about $-2.9$.
The main cost vs error result is shown in Figure~\ref{fig:cost}.

\subsubsection*{Acknowledgements} 
{ We extend our gratitude to H{\aa}kon Hoel for 
identifying the SPDE coupling, which we have used here. }
AJ \& YZ were supported by Ministry of Education AcRF tier 2 grant,
R-155-000-161-112. 
KJHL was supported by  
ORNL LDRD Seed funding grant number 32102582.
AJ and KK were supported by JST CREST grant number JPMJCR14D7, and KK was additionally supported by JSPS KAKENHI grant number 16K00046.

\appendix

\section{Technical Result}

Set 
$$
\tilde{\varphi}(x_{\alpha(1)},\dots,x_{\alpha(k_{\alpha})})  := \sum_{i=1}^{k_{\alpha}'} (-1)^{|\alpha(k_{\alpha})-\alpha(2i)|}\Big\{
\varphi_{\alpha(2i)}(x_{\alpha(2i)}) H_{2i,\alpha}(x_{\alpha(1)},\dots,x_{\alpha(k_{\alpha})}) -
$$
$$
\varphi_{\alpha(2i-1)}(x_{\alpha(2i-1)}) H_{2i-1,\alpha}(x_{\alpha(1)},\dots,x_{\alpha(k_{\alpha})})
\Big\}
$$
where $\varphi_{\alpha(i)}, H_{i,\alpha}$ etc are as in Proposition \ref{prop:main_res}. Also
set 
$$
\bar{\varphi}(x_{\alpha(1)},\dots,x_{\alpha(k_{\alpha})})  :=
G_{\alpha}(x_{\alpha(1)},\dots,x_{\alpha(k_{\alpha})})
\tilde{\varphi}(x_{\alpha(1)},\dots,x_{\alpha(k_{\alpha})}).
$$
We have the following result.

\begin{lem}[Variance of multi-increment]\label{lem:tech_res}
Assume (A\ref{hyp:A}-\ref{hyp:B}). Then for $\varphi:\mathsf{X}\rightarrow\mathbb{R}$ bounded
there exist a $C<+\infty$ independent of $\alpha$ such that
$$
\mathbb{V}\textrm{ar}_{\Pi_{\alpha}}[\tilde{\varphi}(X_{\alpha(1)},\dots,X_{\alpha(k_{\alpha})})]
\leq C D(\alpha)
$$
where $D(\alpha)$ is as in Proposition \ref{prop:main_res}.
\end{lem}

\begin{proof}
Throughout the proof $C$ is a positive and finite scalar constant whose value may change from line-to-line, but does not depend upon $\alpha$.
Set 
\begin{eqnarray*}
Z_{\alpha} & = &
\int_{\bigotimes_{k=1}^{k_{\alpha}} \mathsf{X}_{\alpha(k)}} G_{\alpha}(x_{\alpha(1)},\dots,x_{\alpha(k_{\alpha})}) Q_\alpha(\dif (x_{\alpha(1)},\dots,x_{\alpha(k_{\alpha})})) \\
B_{\alpha} & = & \mathbb{E}_{Q_\alpha}[\tilde{\varphi}(X_{\alpha(1)},\dots,X_{\alpha(k_{\alpha})})]-
\mathbb{E}_{\Pi_{\alpha}}[\tilde{\varphi}(X_{\alpha(1)},\dots,X_{\alpha(k_{\alpha})})] \\
V_{\alpha} & = & \mathbb{E}_{Q_\alpha}[G_{\alpha}(X_{\alpha(1)},\dots,X_{\alpha(k_{\alpha})})
(\tilde{\varphi}(X_{\alpha(1)},\dots,X_{\alpha(k_{\alpha})}) - \\ & &
\mathbb{E}_{Q_\alpha}[\tilde{\varphi}(X_{\alpha(1)},\dots,X_{\alpha(k_{\alpha})})]
)^2]\\
F_{\alpha} & = & \mathbb{E}_{Q_\alpha}[G_{\alpha}(X_{\alpha(1)},\dots,X_{\alpha(k_{\alpha})})
(\tilde{\varphi}(X_{\alpha(1)},\dots,X_{\alpha(k_{\alpha})}) - \\ & &
\mathbb{E}_{Q_\alpha}[\tilde{\varphi}(X_{\alpha(1)},\dots,X_{\alpha(k_{\alpha})})]
)]
\end{eqnarray*}
then
$$
\mathbb{V}\textrm{ar}_{\Pi_{\alpha}}[\tilde{\varphi}(X_{\alpha(1)},\dots,X_{\alpha(k_{\alpha})})]
= \frac{1}{Z_{\alpha}}
\Big[V_{\alpha} + B_{\alpha}^2Z_{\alpha} + 2B_{\alpha}F_{\alpha}\Big].
$$
Note that
$$
B_{\alpha} = \mathbb{E}_{Q_\alpha}[\tilde{\varphi}(X_{\alpha(1)},\dots,X_{\alpha(k_{\alpha})})] -
\frac{1}{Z_{\alpha}}\mathbb{E}_{Q_\alpha}[\bar{\varphi}(X_{\alpha(1)},\dots,X_{\alpha(k_{\alpha})})].
$$
 
(A\ref{hyp:A}) establishes the existence of a $C>0$ such that 
$C^{-1} \leq Z_{\alpha}\leq C$.
Applying also (A\ref{hyp:B}), one has 
$$
|B_{\alpha}| \leq (1+C^{-1})C_{1}(\alpha).
$$
By (A\ref{hyp:B}) and (A\ref{hyp:A}),  
$V_{\alpha}\leq C C_{2}(\alpha)^2$. 
Furthermore, 
by (A\ref{hyp:A}) and Jensen's inequality,
$$|F_{\alpha}|\leq C V_\alpha^{1/2} \leq C C_{2}(\alpha).$$ 
Thus it easily follows that 
$$
\mathbb{V}\textrm{ar}_{\Pi_{\alpha}}[\tilde{\varphi}(X_{\alpha(1)},\dots,X_{\alpha(k_{\alpha})})]
\leq C D(\alpha).
$$
\end{proof}

\section{Analytical solution of the SPDE inverse problem}
\label{app:spde}

Let 
and let ${\bf u}$ denote the concatenated vector such that 
${\bf u}_{j}=u(1/3,t_j)$ for $j=1,\dots , m$, 
${\bf u}_{j}=u(2/3,t_{j-m})$ for $j=m+1,\dots, 2m$,
and ${\bf u}_{2m+1}=u(1/2,T)$.
Then ${\bf u} \sim N(m,\Sigma)$, 
where $m$ and $\Sigma$ are defined 
element-wise in the continuous-time limit by
\[\begin{split}
m_{j} = & \sum_{k=1}^{K_{\max}} e^{(-\pi^2 k^2 + \theta)
(t_j {\bf 1}_{\{j\leq m\}} + t_{j-m} {\bf 1}_{\{j> m\}}+T\delta_{j,2m+1})} u_{k,0} \times \\
& (e_k(1/3) {\bf 1}_{\{j\leq m\}} + e_k(2/3) {\bf 1}_{\{j> m\}} + e_k(1/2)\delta_{j,2m+1}) \, ,
\end{split} \]
and
$$
\Sigma_{ij} = 
\sum_{k}^{K_{\max}} e_k(1/3)^2
\left(\frac{\sigma^2(1-e^{2(\theta-\pi^2k^2)\min\{t_i,t_j\}})}{2(\theta-\pi^2k^2)}\right) 
{\bf 1}_{\{i,j \leq m\}} \ \,
+ 
$$
$$
\sum_{k=1}^{K_{\max}} e_k(2/3)^2 \left(\frac{\sigma^2(1-e^{2(\theta-\pi^2k^2)\min\{t_{i-20},t_{j-20}\}})}{2(\theta-\pi^2k^2)}\right) {\bf 1}_{\{m<i,j\leq 2m\}} \ \, +
$$
$$
\sum_{k=1}^{K_{\max}} e_k(1/3)e_k(2/3)\left(\frac{\sigma^2(1-e^{2(\theta-\pi^2k^2)\min\{t_i,t_{j-20}\}})}{2(\theta-\pi^2k^2)}\right) {\bf 1}_{\{i\leq m < j \leq 2m\}} \ \, + 
$$
$$
\sum_{k=1}^{K_{\max}} e_k(2/3)e_k(1/3)\left(\frac{\sigma^2(1-e^{2(\theta-\pi^2k^2)\min\{t_{i-20},t_{j}\}})}{2(\theta-\pi^2k^2)}\right) {\bf 1}_{\{j \leq m < i \leq 2m\}} \ \, + 
$$
$$
\sum_{k=1}^{K_{\max}} e_k(1/2)e_k(1/3)\left(\frac{\sigma^2(1-e^{2(\theta-\pi^2k^2)t_i})}{2(\theta-\pi^2k^2)}\right) {\bf 1}_{\{i \leq m, j=2m +1\}} \ \, + 
$$
$$
\sum_{k=1}^{K_{\max}} e_k(1/2)e_k(1/3)\left(\frac{\sigma^2(1-e^{2(\theta-\pi^2k^2)t_j})}{2(\theta-\pi^2k^2)}\right) {\bf 1}_{\{j \leq m, i=2m +1\}} \ \, + 
$$
$$
\sum_{k=1}^{K_{\max}} e_k(1/2)e_k(2/3)\left(\frac{\sigma^2(1-e^{2(\theta-\pi^2k^2)t_{i-m}})}{2(\theta-\pi^2k^2)}\right) {\bf 1}_{\{m< i \leq 2m, j=2m +1\}} \ \, + 
$$
$$
\sum_{k=1}^{K_{\max}} e_k(1/2)e_k(2/3)\left(\frac{\sigma^2(1-e^{2(\theta-\pi^2k^2)t_{j-m}})}{2(\theta-\pi^2k^2)}\right) {\bf 1}_{\{m< j \leq 2m, i=2m +1\}} \ \, + 
$$
$$
\sum_{k=1}^{K_{\max}} e_k(1/2)e_k(1/3)\left(\frac{\sigma^2(1-e^{2(\theta-\pi^2k^2)T})}{2(\theta-\pi^2k^2)}\right) {\bf 1}_{\{i=j=2m +1\}} \ \, .
$$
A similar expression can be obtained for the time-discretized version, 
but for our purposes, i.e. as a ground truth, this will be sufficient.

Given the additive Gaussian noise assumption on the observations,
the posterior is known 
and it is given by 
${\bf u}|y \sim N(\widehat{m}, \widehat{\Sigma})$,
where, letting $H=(I_{2m},{\bf 0}_{2m \times 1})$,  
  \begin{eqnarray}\label{eq:posteriorX}
\widehat{m} &=& 
\widehat\Sigma \left ( \left (\Sigma\right )^{-1} m + \frac{1}{\tau^2}H^\top y\right) \ , \\
\widehat{\Sigma} &=& \left (\frac{1}{\tau^2} H^\top H 
+ \left (\Sigma\right )^{-1} \right )^{-1} \ .
\end{eqnarray}

\end{document}